\documentclass[times,12pt,draftcls,onecolumn]{IEEEtran}
\usepackage{amsfonts,amssymb}
\usepackage{graphicx}

\newcommand{\vx}{\mbox{${\bf x}$}}

\newcommand{\vs}{\mbox{${\bf s}$}}

\newcommand{\vn}{\mbox{${\bf n}$}}

\newcommand{\mA}{\hbox{{\bf A}}}

\newcommand{\mB}{\hbox{{\bf B}}}

\newcommand{\mD}{\hbox{{\bf D}}}

\newcommand{\mE}{\hbox{{\bf E}}}

\newcommand{\mF}{\hbox{{\bf F}}}
\newcommand{\mH}{\hbox{{\bf H}}}

\newcommand{\mI}{\hbox{{\bf I}}}

\newcommand{\mP}{\hbox{{\bf P}}}
\newcommand{\mQ}{\hbox{{\bf Q}}}

\newcommand{\gs}{\sigma}

\def\bm#1{\mbox{\boldmath $#1$}}

\newcommand{\mgD}{\mbox{$\bm \Delta$}}

\newtheorem{theorem}{Theorem}[section]
\newtheorem{lemma}[theorem]{Lemma}

\newtheorem{prop}{Proposition}[section]
\newtheorem{claim}{Claim}[section]

\newtheorem{definition}{Definition}[section]
\newtheorem{question}{Question}[section]
\newtheorem{coro}{Corollary}[section]

\newcommand{\beq}{\begin{equation}}
\newcommand{\eeq}{\end{equation}}
\newcommand{\bea}{\begin{array}}
\newcommand{\ena}{\end{array}}
\newcommand{\bds}{\begin {itemize}}
\newcommand{\eds}{\end {itemize}}
\newcommand{\bdf}{\begin{definition}}
\newcommand{\blm}{\begin{lemma}}
\newcommand{\edf}{\end{definition}}
\newcommand{\elm}{\end{lemma}}
\newcommand{\bthm}{\begin{theorem}}
\newcommand{\ethm}{\end{theorem}}
\newcommand{\bprp}{\begin{prop}}
\newcommand{\eprp}{\end{prop}}
\newcommand{\bcl}{\begin{claim}}
\newcommand{\ecl}{\end{claim}}
\newcommand{\bcr}{\begin{coro}}
\newcommand{\ecr}{\end{coro}}
\newcommand{\bquest}{\begin{question}}
\newcommand{\equest}{\end{question}}

\newcommand{\larrow}{{\larrow}}

\newtheorem{Theorem}[theorem]{{\bf Main Theorem}}

\newtheorem{corollary}[theorem]{Corollary}

\newtheorem{'thm'}{"Theorem"}

\newtheorem{remark}[theorem]{Remark}

\newcommand{\qed}{\nobreak \ifvmode \relax \else
    \ifdim\lastskip<1.5em \hskip-\lastskip
    \hskip1.5em plus0em minus0.5em \fi \nobreak
    \vrule height0.75em width0.5em depth0.25em\fi}

\begin{document}
\author{Eitan Sayag$^{1,2}$, Amir Leshem$^{1,3}$ {\it Senior member, IEEE} \\ and \\
Nicholas D. Sidiropoulos$^4$ {\it Senior member, IEEE}
\thanks{$^1$ School of Engineering, Bar-Ilan University, 52900, Ramat
Gan, Israel. $^2$ Dept. of Math. Ben Gurion University, Beer-Sheva, Israel. $^3$ Circuit and Systmes, Faculty of EEMCS, Delft University of Technology, Delft, The Netherlands.
$^4$ Dept. of ECE, Technical
University of Crete, Greece. This research was partially supported
by the EU-FP6 IST, under contract no. 506790, and by the Israeli
ministry of trade and commerce as part of Nehusha$\backslash$iSMART
project. Conference version of part of this work appeared in {\it
Proc. IEEE ICASSP 2008}, Mar. 30 - Apr. 4, 2008, Las Vegas,
Nevada.}}
\title{Finite Word Length Effects on Transmission Rate in Zero Forcing Linear Precoding for Multichannel DSL}
\maketitle

\begin{abstract}
Crosstalk interference is the limiting factor in transmission over
copper lines. Crosstalk cancelation techniques show great potential
for enabling the next  leap in DSL transmission rates. An important
issue when implementing crosstalk cancelation techniques in hardware
is the effect of finite world length on performance. In this paper
we provide an analysis of the performance of linear zero-forcing
precoders, used for crosstalk compensation, in the presence of
finite word length errors. We quantify analytically the trade off
between precoder word length and transmission rate degradation. More
specifically, we prove a simple formula for the transmission rate
loss as a function of the number of bits used for precoding, the
signal to noise ratio, and the standard line parameters. We
demonstrate, through simulations on real lines, the accuracy of our
estimates. Moreover, our results are stable in the presence of
channel estimation errors. Finally, we show how to use these
estimates as a design tool for DSL linear crosstalk precoders. For
example, we show that for standard VDSL2 precoded systems, 14 bits
representation of the precoder entries results in capacity loss
below $1 \%$ for lines over 300m. 
\end{abstract}



\section{Introduction}

DSL systems are capable of delivering high data rates over copper
lines. A major problem of DSL technologies is the electromagnetic
coupling between the twisted pairs within a binder group. Reference
\cite{ginis2002} and the recent experimental studies in
\cite{karipidis2005b,karipidis2005a} have demonstrated that {\em
vectoring} and {\em crosstalk cancelation} allow a significant
increase of the data rates of DSL systems. In particular, linear
precoding has recently drawn considerable attention
\cite{cendrillon2006,cendrillon2007} as a natural method for
crosstalk precompensation as well as crosstalk cancelation in the
receiver. In \cite{karipidis2005b,karipidis2005a} it is shown that
optimal cancelation achieves capacity boost ranging from $2 \times$
to $4 \times$, and also substantially reduces per-loop capacity
spread and outage, which are very important metrics from an
operator's perspective. References
\cite{cendrillon2007,cendrillon2004a} advocate the use of a
diagonalizing precompensator, and demonstrate that, without
modification of the Customer Premise Equipment (CPE), one can obtain
near optimal performance. Recent work in
\cite{leshem2004b,leshem2007a} has shown that a low-order truncated
series approximation of the inverse channel matrix affords
significant complexity reduction in the computation of the precoding
matrix. Implementation complexity (i.e., the actual multiplication
of the transmitted symbol vector by the precoding matrix) remains
high, however, especially for multicarrier transmission which
requires one matrix-vector multiplication for each tone. Current
advanced DSL systems use thousands of tones. In these conditions,
using minimal word length in representing the precoder matrix is
important. However, using coarse quantization will result in
substantial rate loss. The number of quantization bits per matrix
coefficient is an important parameter that affects the system's
performance - complexity trade-off, which we focus on in this paper.
We provide closed form sharp analytic bounds on the absolute and
relative transmission rate loss. We show that both absolute and
relative transmission loss decay exponentially as a function of the
number of quantizer bits and provide explicit bounds for the loss in
each tone. Under analytic channel models as in
\cite{werner91,karipidis2006a} we provide refined and explicit
bounds for the transmission loss across the band and compare these
to simulation results. This explicit relationship between the number
of quantizer bits and the transmission rate loss due to quantization
is a very useful tool in the design of practical systems.

The structure of the paper is as follows. In section II, we
present the signal model for a precoded discrete multichannel
system and provide a model for the precoder errors we study. In
section III, a general formula for the transmission loss of a
single user is derived. In section IV we focus on the case of full
channel state information where the rate loss of a single user
results from quantization errors only. Here we prove the main
result of the paper, Theorem \ref{main theorem}. We provide
explicit bounds on the rate loss under an analytic model for the
transfer function as in \cite{werner91}. We also study a number of
natural design criteria. In section V we provide simulation
results on measured lines, which support our analysis. Moreover,
we show through simulation that our results are valid in the
presence of measurement errors. The appendices provide full
details of the mathematical claims used in the main text.

\section{Problem Formulation}

\subsection{Signal model}
In this section we describe the signal model for a precoded discrete
multitone (DMT) system. We assume that the transmission scheme is
Frequency Division Duplexing (FDD), where the upstream and the
downstream transmissions are performed at separate frequency bands.
Moreover, we assume that all modems are synchronized. Hence, the
echo signal is eliminated, as in \cite{ginis2002}, and the received
signal model at frequency $f$ is given by
\begin{equation}\label{signal_model clean}
\vx(f)=\mH(f) \vs(f)+\vn(f),
\end{equation}
where $\vs(f)$ is the vectored signal sent by the optical network
unit (ONU), $\mH(f)$ is a $p \times p$ matrix representing the
channels, $\vn(f)$ is additive Gaussian noise, and $\vx(f)$
(conceptually) collects the signals received by the individual
users. The users estimate rows of the channel matrix $\mH(f)$, and
the ONU uses this information to send $\mP(f) \vs(f)$ instead of
$\vs(f)$. This process is called {\it crosstalk pre-compensation}.
In general such a mechanism yields
\begin{equation}\label{signal_model_precoding}
\vx(f)=\mH(f) \mP(f) \vs(f)+\vn(f).
\end{equation}
Denote the diagonal of $\mH(f)$ by $\mD(f)=diag(\mH(f))$ and let
$\mP(f)=\mH(f)^{-1} \mD(f)$ as suggested in \cite{cendrillon2007}.
With this we have
\begin{equation}\label{signal_model_diag}
\vx(f)=\mD(f) \vs(f)+\vn(f),
\end{equation}
showing that the crosstalk is eliminated. Note that with
$\mF(f)=\mH(f)-\mD(f)$ we have the following formula for the matrix
$\mP(f)$
\begin{equation}\label{ideal precoder}
\mP(f)=\left(\mI+\mD^{-1}(f)\mF(f)\right)^{-1}.
\end{equation}
Following \cite{cendrillon2007} we assume that the matrices ${\bf
H}(f)$ are row-wise diagonally dominant, namely that
\begin{equation}
\| h_{ii} \| >> \| h_{ij}\|, \forall i \neq j.
\end{equation}
In fact, motivated in part by Gersgorin's theorem \cite{horn} we
propose the parameter $r(\mH)$
\begin{equation}\label{row_dominance_parameter}
r(\mH)=\max_{1 \le i \le N} \left(\frac{\sum_{j \neq i}
|h_{ij}|}{|h_{ii}|}\right),
\end{equation}
as a measure for the dominance. In most downstream scenarios the
parameter $r$ is indeed much smaller than 1. We emphasize that
typical downstream VDSL channels are row-wise diagonally dominant
even in mixed length scenarios as demonstrated in
\cite{leshem2007a}.

\subsection{A model for precoder errors}

In practical implementations, the entries of the precoding matrix
$\mP$ will be quantized. The number of quantizer bits used is
dictated by complexity and memory considerations. Indeed,
relatively coarse quantization of the entries of the precoder
$\mP$ allows significant reduction of the time complexity and the
amount of memory needed for the precoding process. The key problem
is to determine the transmission rate loss of an individual user
caused by such quantization. Another closely related problem is
the issue of robustness of linear precoding with respect to errors
in the estimation of the channel matrix. The mathematical setting
for both is that of error analysis. Let
\begin{equation}\label{precoder}
\mP=(\mI+\mD^{-1}\mF+\mE_{1})^{-1}+\mE_{2},
\end{equation}
where
\begin{itemize}
\item $\mE_{1}$ models the {\it relative error in quantizing or
measuring} the channel matrix $\mH$, and \item $\mE_{2}$ models
the {\it errors caused by quantizing} the precoder $\mP$.
\end{itemize} The problem is to determine the capacity of the
system, and the capacity of each user, in terms of the system
parameters and the statistical parameters of the errors. Note that
equation (\ref{precoder}) captures three types of errors:
errors in the estimation of $\mH$, quantization errors in the
representation of $\mH$, and quantization errors in the
representation of the precoder $\mP$.

Our focus will be in the study of the effect of quantization errors
in the representation of the precoder on the capacity of an
individual user. Nevertheless, the estimation errors resulting from
measuring the channel cannot be ignored.
We will show that the analysis of quantization errors and
estimation errors can be dealt separately (see remark \ref{rmkref}
after lemma \ref{lemma of section 3}). This allows us to carry
analysis under the assumption of perfect channel information.
Then, we show in simulations that when the estimation errors in
channel measurements are reasonably small, our analytical bounds
remain valid.

\subsection{
System Model}

We now list our assumptions regarding the errors
$\mE_{1},\mE_{2}$, the power spectral density of the users, and
the behavior of the channel matrices.

{\bf Perfect CSI: } Perfect Channel Information. Namely,
\begin{equation}\label{assume Perfect CSI}
\mE_{1}(f)=0, ~~~ \forall f.
\end{equation}

{\bf Quant$(2^{-d})$: } The quantization error of each matrix
element of the precoder is at most $2^{-d}$. Namely,

\begin{equation}\label{assume Quantization}
|\mE_{2}(f)_{i,j}| \leq 2^{-d}, ~~~ \forall f, \forall i,j.
\end{equation}

{\bf DD:} The channel matrices are row-wise diagonally dominant.

\begin{equation}\label{assume DD}
r(\mH(f)) \leq 1,  ~~~ \forall f.
\end{equation}


{\bf SPSD:}  The Power Spectral Density (PSD) of all the users of
the binder is the same. Namely, we assume that for some fixed
unspecified function $P(f)$ we have:
\begin{equation}\label{assume SPSD}
P_{i}(f)=P(f), ~~~ \forall i.
\end{equation}

The main result of the paper, Theorem \ref{main theorem} is based
on assumptions (\ref{assume Perfect CSI}), (\ref{assume
Quantization}), (\ref{assume DD}), (\ref{assume SPSD}).

Assumption {\bf SPSD} can be lifted, as shown in section
\ref{lifting of P} (appendix H). For the sake of clarity we
present only the simplified result in the body of the paper.

 In order to obtain sharp analytic estimates on the transmission
loss in actual DSL scenarios we need to incorporate some of the
properties of the channel matrices of DSL channels into our model.
In particular, we will assume

{\bf Werner Channel model:} The matrix elements of the channel
matrices $\mH(f)$ behave as in the model of \cite{werner91}.
Namely, following \cite{werner91} we assume the following model
for insertion loss
\begin{equation}\label{assume_werner1}
|\mH^{IL}(f,\ell)|^{2}=e^{-2 \alpha \ell \sqrt{f}}
\end{equation}
where $\ell$ is the  DSL loop length (in meters), $f$ is the
frequency in Hz, and $\alpha$ is a parameter that depends on the
cable type. Furthermore, crosstalk is modeled as
\begin{equation}\label{assume_werner2}
|\mH^{FEXT}(f,\ell)|^{2}=K(\ell)f^{2}|\mH^{IL}(f,\ell)|^{2}
\end{equation}
Here $K(\ell)$ is a random variable studied in
\cite{karipidis2006a}. The finding is that $K(\ell)$ is a
log-normal distribution with expectation, denote there
$c_{1}(\ell)$, that increase linearly with $\ell$.







An additional assumption that we will make concerns the behavior
of the row dominance of the  channel matrices $\mH(f,\ell)$.

{\bf Sub linear row dominance:}
\begin{equation}\label{assume SLRD}
r(\mH(f,\ell)) \leq \gamma_{1}(\ell)+\gamma_{2}(\ell) f
\end{equation}
 Where
$\gamma_{2}(\ell)=O(\sqrt{\ell})$.

\begin{remark}
Note that
$$\frac{|\mH^{FEXT}(\ell,f)|}{|\mH^{IL}(f,\ell)|}=\sqrt{K(\ell)}f.$$
The sub-linearity in $f$ follows by studying $r(\mH(\ell,f))$ in
terms of $p^{2}$ random variables behaving as $K(\ell)$.
\end{remark}




\subsection{Justification of the assumptions}


{\bf Perfect CSI}  is plausible due to the quasi-stationarity of
DSL systems (long coherence time), which allows us to estimate the
channel matrices at high precision.

{\bf Quant$(2^{-d})$} is a weak assumption on the type of the
quantization process. Informally it is equivalent to an assumption
on the number of bits used to quantize an entry in the channel
matrix. In particular, our analysis of the capacity loss will be
independent of the specific quantization method and our results are
valid for any technique that quantizes matrix elements with bounded
errors.

Assumption {\bf DD} reflects the diagonal dominance of DSL channels.
While linear precoding may result in power fluctuations, the
diagonal dominance property of DSL channel matrices makes these
fluctuations negligible within 3.5dB fluctuation allowed by PSD
template (G993.2). For example if the row dominance is up to 0.1 the
effect of precoding on the transmit powers and spectra will be at
most 1dB.

Assumption {\bf SPSD} (see (\ref{assume SPSD})) is justified in a
system with ideal full-binder precoding, where each user will use
the entire PSD mask allowed by regulation. Note that in
\cite{karipidis2005a} it is shown that DSM3 provides significant
capacity gains only when almost all pairs in a binder are
coordinated. Thus the equal transmit spectra assumption is
reasonable in these systems. However we also provide in section
\ref{lifting of P} (appendix H) a generalization of the main result
to a setting in which this assumption is not satisfied.

Assumption {\bf Werner Channel model} does not need justification
whereas our last assumption, {\bf sub-linear row dominance} was
verified on measured lines \cite{karipidis2005a} and can also be
deduced analytically from Werner's model.  In practice, the type of
fitting required to obtain $\gamma_{1}(\ell), \gamma_{2}(\ell)$ from
measured data is simple and can be done efficiently. Moreover, the
line parameters tabulated in standard (e.g., R,L,C,G parameters of
the two port model), together with the $99\%$ worst case power sum
model used in standards \cite{ETSI_cables}, provide another way of
computing the constants $\gamma_{1}(\ell), \gamma_{2}(\ell)$.

\section{A General Formula For Transmission Loss}
The purpose of this section is to provide a general formula for
the transmission rate loss of a single user, resulting from errors
in the estimated channel matrix as well as errors in the precoder
matrix.
First, we develop a useful expression for the equivalent channel in
the presence of errors. This is given in formula (\ref{signal model
real}). Next, a formula for the transmission loss is obtained
(\ref{trans loss formula}). The formula compares the achievable rate
of a communication system using an ideal ZF precoder as in
(\ref{ideal precoder}) versus that of a communication system whose
precoder is given by (\ref{precoder}). This formula is the key to
the whole paper. Note that we use a gap analysis as in
\cite{cioffi95a,cioffi95b}. A useful corollary in the form of
formula (\ref{The mother of all formulas}) is derived. This will be
used in the next section to obtain bounds on capacity loss due to
quantization.

 Let
$\mH(f)=\mD(f)+\mF(f)$ be a decomposition of the channel matrix at
a given frequency to diagonal and non-diagonal terms. Thus
$\mD(f)$ is a diagonal matrix whose diagonal is identical to the
diagonal of $\mH(f)$. Also we let $SNR_{i}(f)$ be the signal to
noise ratio of the $i$-th receiver at frequency $f$
\begin{equation}\label{sn}
SNR_{i}(f)=\frac{P_{i}(f)|d_{i,i}(f)|^{2}}{E|n_{i}(f)|^{2}}.
\end{equation}
In this formula $P_{i}(f)$ is the power spectral density (PSD) of
the $i$-th user at frequency $f$, and $n_{i}(f)$ is the associated
noise term. We denote
\beq
\gs_{n_i}^2(f)=E|n_{i}(f)|^2.
\eeq

\subsection{A formula for the equivalent channel in the presence of errors}
We first derive a general formula for the equivalent signal model.
The next lemma provides a useful reformulation of the signal model
in (\ref{signal_model_precoding}):
\begin{lemma}\label{lemma of section 3}
The precoded channel (\ref{signal_model_precoding}) with precoder
as in (\ref{precoder}) is given by
\begin{equation}\label{signal model real}
\vx(f)= \mD(f) \vs(f)+ \mD(f)  \mgD(f)
 \vs(f) + \vn(f),
\end{equation}
 with
\begin{equation}\label{delta}
\mgD(f) = (\mI+{\mD^{-1}(f)}{\mF(f)}) \mE_{2}(f)
-\mE_{1}(f)(\mI+\mD^{-1}(f)\mF(f)+\mE_{1}(f))^{-1}.
\end{equation}
\end{lemma}
The proof is deferred to appendix A (section \ref{proof of lemma
1}).
\begin{remark}\label{rmkref}
For our analysis, we will assume that $\mE_{1}(f)=0$, in which
case the formula for the matrix $\mgD$ simplifies to
\begin{equation}\label{delta simple}
\mgD(f) = (\mI+{\mD^{-1}(f)}{\mF(f)}) \mE_{2}(f). \end{equation}

The relevance of the formula (\ref{delta}) for the experimental
part of the paper (where $\mE_{1}(f)$ is not assumed to be zero)
is explained in the next remark.
\end{remark}

\begin{remark}
In formula (\ref{trans loss formula}) below we show that the
impact of the errors $\mE_{1}(f)$ and $\mE_{2}(f)$ on the
transmission loss of a user can be computed from the matrix
$\mgD$. Thus, an important consequence of the lemma is that the
effect on transmission loss due to estimation errors (encoded in
the matrix $\mE_{1}(f)$) and due to quantization errors (encoded
in the matrix $\mE_{2}(f)$) can be studied separately as they
contribute to different terms in the above expression for $\mgD$.
\end{remark}

\subsection{Transmission Loss of a Single User}

Consider a communication system as defined in
(\ref{signal_model_diag}) and denote by
 $B$ the frequency band of the system.
We let $SNR_{i}(f)$ be as in (\ref{sn}) and let $\Gamma$
 be the Shannon Gap comprising modulation loss, coding gain and noise margin.
Let $R_{i}$ be the achievable transmission rate  of the $i$-th
user in the system defined in (\ref{signal_model_diag}). Recall
that in such a system the crosstalk is completely removed and
therefore
\begin{equation}\label{Transmission Rate}
 R_{i}=\int_{f \in B} \log_{2}(1+\Gamma^{-1}SNR_{i}(f))df.
\end{equation}
Let
\begin{equation}
 R_{i}(f)=\log_{2}(1+\Gamma^{-1}SNR_{i}(f))
\end{equation}
be the transmission rate at frequency $f$ (formally, it is just
the density of that rate).
Let $\tilde{R}_{i}(f)$ be the transmission rate at frequency $f$
of the $i$-th user, when the precoder in (\ref{precoder}) is used.
We note that while $R_{i}(f)$ is a number, the quantity
$\tilde{R}_{i}(f)$ depends on the random variables ${\bf E}_{1},
{\bf E}_{2}$ and hence is itself a random variable. Let
$\tilde{R}_{i}$ be the transmission rate of the $i$-th user for
the equivalent system in (\ref{signal model real}). Thus,
\begin{equation}\label{R tilde}
 \tilde{R}_{i}=\int_{f \in B} \tilde{R}_{i}(f)df.
\end{equation}
By equation (\ref{signal model real}), the $i$-th user receives
\begin{equation}\label{system with errors_a}
x_{i}(f)=d_{i,i}(f)s_{i}(f)+d_{i,i} \sum_{j=1}^{p}
\Delta_{i,j}(f)s_{j}(f)+n_{i}(f)=d_{i,i}(f)(1+\Delta_{i,i}(f))s_{i}(f)+N_{i}(f)
\end{equation}
where $N_{i}(f)=d_{i,i}(f) \sum_{j \ne i}^{p}
\Delta_{i,j}(f)s_{j}(f)+n_{i}(f)$. Assuming Gaussian signaling
i.e. that all $s_{i}(f)$ are Gaussian we conclude that $N_{i}(f)$
is Gaussian. A similar conclusion is valid in the case of a large
number of users, due to the Central Limit Theorem. In practice,
the Gaussian assumption is a good approximation even for a modest
number of (e.g., 8) users. Recall also that Gaussian signaling is
the optimal strategy in the case of exact channel knowledge.
Therefore, we can use the capacity formula for the Gaussian
channel, even under precoder quantization errors. \bdf The
transmission loss $L_{i}(f)$ of the $i$-th user at frequency $f$
is given by
\begin{equation}\label{trans loss def}
L_{i}(f)=R_{i}(f)-\tilde{R}_{i}(f).
\end{equation}
The {\bf total loss} of the $i$-th user is
\begin{equation}\label{trans loss global}
L_{i}=\int_{f \in B} L_{i}(f)df.
\end{equation}
\edf

 We are ready to deduce a formula for the rate loss
of the $i$-th user as a result of the non-ideal precoder in
(\ref{signal model real}). Our result will be given in terms of the
matrix $\mgD$. Recall that $\mgD$ generally depends on both precoder
quantization errors $\mE_{2}$ and estimation errors
 $\mE_{1}$.

Denote by $\mgD_{i,j}$ the $(i,j)$-th element of the matrix $\mgD$
and let \beq \delta_{i}(f)=\Gamma \sum_{j \ne i}
\frac{P_{j}(f)}{P_{i}(f)} |\Delta_{i,j}(f)|^{2}. \eeq
 Let

\begin{equation}\label{a number}
a_{i}(f)=\delta_{i}(f)\Gamma^{-1}SNR_{i}(f)=\sum_{j \ne i}
\frac{P_{j}(f)}{P_{i}(f)} |\mgD_{i,j}(f)|^{2} SNR_{i}(f),
\end{equation}


\begin{equation}\label{q number}
q_{i}(\mgD,f)=\frac{|1+\mgD_{i,i}(f)|^{2}}{a_{i}(f)+1},
\end{equation}
 and
 \begin{equation}\label{k number}
  k_{i}(f)=\frac{\Gamma^{-1}SNR_{i}(f)}{\Gamma^{-1}SNR_{i}(f)+1}.
 \end{equation}

Note that $a_{i}(f)$ and hence $q_{i}(\mgD,f)$ are independent of
the Shannon gap $\Gamma$. The next lemma provides a formula for the
exact transmission rate loss due to the errors modeled by the
matrices $\mE_{1}$ and $\mE_{2}$. The result is stated in terms of
quantities $q(\mgD,f)$ and the effective signal to noise ratio,
$\Gamma^{-1}SNR_{i}(f)$.

\begin{lemma}\label{main prop}
Let $\mH(f)$ be the channel matrix at frequency $f$ and let
$\mE_{1},\mE_{2}$ be the estimation and quantization errors,
respectively as in (\ref{precoder}). Let $L_{i}(f)$ be the loss in
transmission rate of the $i$-th user defined in (\ref{trans loss
def}). Then

\begin{equation}\label{trans loss formula}
 L_{i}(\mgD,f) =
 -\log_{2}\left
 (1-k_{i}(f)
 (1-q_{i}(\mgD,f))\right),
 \end{equation}

 where $q_{i}(\mgD,f)$ is given in (\ref{q number}) and $k_{i}(f)$ is
 given in (\ref{k number}).

In particular, if $\mgD_{i,i}(f)=-1$ the transmission loss is
$\log_{2}(1+\Gamma^{-1}SNR_{i}(f))$, where $SNR_{i}(f)$ is defined
in (\ref{sn}). Finally, if $\Delta_{i,i}(f) \ne -1$ we have
\begin{equation}\label{basic inq}
L_{i}(\mgD,f) \leq
Max\left(0,\log_{2}\left(\frac{1}{q_{i}(\mgD,f)}\right)\right)
\end{equation}

\end{lemma}

The proof of this lemma is deferred to appendix B (section
\ref{proof of lemma 2}).

To formulate a useful corollary we introduce the quantities:
\begin{equation}
M_{i}(f)=max_{j \ne i} \frac{P_{j}(f)}{P_{i}(f)}
\end{equation}

\begin{equation}
t_{i}(f)=max_{1 \leq j \leq n}|\mgD_{i,j}|
\end{equation}

\begin{corollary}
Let $\mH(f)$ be the $p \times p$ channel matrix at frequency $f$ and
let $\mE_{1}(f)$, $\mE_{2}(f)$ be the estimation and quantization
errors respectively as in (\ref{precoder}). Let $L_{i}(f)$ be the
transmission rate loss of the $i$-th user defined in (\ref{trans
loss def}). Assume that $t_{i}(f)< 1$. Then
\begin{equation}\label{The mother of all formulas}
L_{i}(\mgD,f) \leq \log_{2}\left(
                          \frac{1+(p-1)M_{i}(f)t_{i}^{2}(f)SNR_{i}(f)}{(1-t_{i}(f))^{2}}
                          \right)
\end{equation}
\end{corollary}

{\it Proof:}
By (\ref{a number}) we have
\begin{equation}
a_{i}(f)=\sum_{j \ne i} \frac{P_{j}(f)}{P_{i}(f)}
|\mgD_{i,j}(f)|^{2} SNR_{i}(f) \leq
M_{i}(f)t_{i}(f)^{2}(p-1)SNR_{i}(f)
\end{equation}

\beq 1+a_{i}(f) \leq 1+(p-1)M_{i}(f)t_{i}(f)^{2}SNR_{i}(f) \eeq
Since $|\mgD_{i,i}(f)| \leq t_{i}(f)$ we get
\begin{equation}
|1+\mgD_{i,i}(f)|^{2} \geq (1 - t_{i}(f))^{2}
\end{equation}
Thus by (\ref{q number}) we have
\begin{equation}
\frac{1}{q_{i}(\mgD,f)} = \frac{a_{i}(f)+1}{|1+\mgD_{i,i}(f)|^{2}}
\leq \frac{1+(p-1)M_{i}(f)t_{i}^{2}(f)SNR_{i}(f)}{(1-t_{i}(f))^{2}}
\end{equation}
Notice that the right hand side is larger than one and
using (\ref{basic inq}) of the previous lemma the proof is
complete.
\begin{remark}
We note that under simplifying assumptions, such as assumption
SPSD (see (\ref{assume SPSD})) the above formula reduces to
\begin{equation}\label{The mother of all formulas_2}
L_{i}(\mgD,f) \leq \log_{2}(1+(p-1)t_{i}^{2}(f)SNR_{i}(f))-
                 2\log_{2}(1-t_{i}(f))
\end{equation}

Under the assumption {\bf Perfect CSI}, we have $\mgD(f) =
(\mI+{\mD^{-1}(f)}{\mF(f)}) \mE_{2}(f)$ and since we further
assumed that the channel matrices $\mH(f)$ are row-wise diagonally
dominant we see that $\mgD(f) \approx \mE_{2}(f)$. Thus, $t_{i}(f)
\approx 2^{-d}$ and we obtain a bound of the form

\begin{equation}\label{The mother of all formulas_3}
L_{i}(\mgD,f) \le \log_{2}(1+(p-1)SNR_{i}(f)2^{-2d})-
                 2\log_{2}(1-2^{-d})
\end{equation}
For a statement of a bound of this form see formula (\ref{bnd in
tone}) of Theorem \ref{main theorem} below.
\end{remark}

\section{Transmission rate Loss Resulting
from Quantization Errors in the Precoder}

In the ZF precoder studied earlier we can assume without loss of
generality that the entries are of absolute value less than one.
Each of these values is now represented using $2d$ bits ($d$ bits
for the real part and $d$ bits for the imaginary part, not including
the sign bit). We first consider an ideal situation in which we have
perfect channel estimation.
\subsection{Transmission Loss with Perfect Channel Knowledge}
Consider the case where $\mE_{1}={\bf 0}$ and the quantization
error is given by an arbitrary matrix $\mE_{2}$ with the property
that each entry is a complex number with real and imaginary parts
bounded in absolute value by $2^{-d}$. We will not make any
further assumptions about the particular quantization method
employed and we will provide upper bounds for the capacity loss.
We do not assume any specific random model for the values of
$\mE_{2}$ because we are interested in obtaining absolute upper
bounds on capacity loss.

The following theorem describes the transmission rate loss
resulting from quantization of the precoder.
\begin{Theorem} \label{main theorem}
Let $\mH(f)$ be the channel matrix of $p$ twisted pairs at
frequency $f$, and $r(f)=r(\mH(f))$ as in
(\ref{row_dominance_parameter}). Assume {\bf Perfect CSI}
(\ref{assume Perfect CSI}), {\bf Quant$(2^{-d})$} (\ref{assume
Quantization}), {\bf SPSD} (\ref{assume SPSD}), and that the
precoder $\mP(f)$ is quantized using $d \geq
\frac{1}{2}+\log_2(1+r(f))$ bits.
The transmission rate loss of
the $i$-th user at frequency $f$ due to quantization is bounded by
\begin{equation}\label{bnd in tone}
 L_{i}(d,f) \leq \log_{2}(1+\gamma(d,f)SNR_{i}(f))-
2\log_{2}(1-v(f)2^{-d}), \end{equation} where \beq
\gamma(d,f)=2(p-1)(1+r(f))^{2}2^{-2d}
 \eeq and
\beq  v(f)=\sqrt{2}(1+r(f)). \eeq Furthermore, suppose $d \geq
\frac{1}{2}+\log_2(1+r_{max})$ with \beq r_{max}=max_{f \in
B}(r(\mH(f)). \eeq

Then the transmission loss in the band $B$ is at most
\begin{equation}\label{int main formula}
\int_{f \in B} \log_{2}(1+\gamma(d)SNR_{i}(f))df - 2|B|
\log_{2}(1-(1+r_{max})2^{-d+0.5}),
\end{equation}
where $|B|$ is the total bandwidth, \beq
\gamma(d)=2(1+r_{max})^{2}(p-1)2^{-2d},
 \eeq
\end{Theorem}


The proof of the theorem is deferred to section \ref{Appendix proof
of main theorem} (appendix C).

We now record some useful corollaries of the theorem illustrating
its value.

\begin{corollary}\label{cor one tone}
The transmission rate loss $L_{i}(\mgD,f)$, due to quantization of
the precoding matrix by $d$ bits is bounded by:
\begin{equation}\label{one tone bound}
L_{i}(\mgD,f) \leq
 \log_{2}(1+\gamma(d,f)SNR_{i}(f))
 -2\log_{2}(1-v(f)2^{-d})
\end{equation}
where $\gamma(d,f)=2(p-1)(1+r(f))^{2}2^{-2d}$ and
$v(f)=\sqrt{2}(1+r(f))$. If $r(f) \leq 1$, a simplified looser
bound is given by
\begin{equation}
\label{simple_bound} L_{i}(\mgD,f) \leq
2^{-d+3.5}+\log_{2}(1+8(p-1)SNR_{i}(f)2^{-2d})
\end{equation}
\end{corollary}

For the derivation of the first inequality see (\ref{main
inequality}) in section \ref{Appendix proof of main theorem}. The
simplified bound is based on the estimate $-\log_{2}(1-z) \leq 2z$
valid for $0 \leq z \leq 0.5 $.

The next result is of theoretical value. It describes the
asymptotic behavior of $L_{i}(d)$ for very large $d$.

\begin{corollary}
Under the assumptions of the theorem and assuming that $r_{max} \leq
1$:
$$L_{i}(d)=O(2^{-d}).$$ More precisely, we have
$$L_{i}(d) = \theta\left(\frac{\sqrt{32}}{\ln(2)} 2^{-d}B\right).$$
\end{corollary}

\begin{remark}
By definition, $f(n)=\theta(g(n))$ if and only if $$lim_{n \to
\infty} \frac{f(n)}{g(n)}=1$$
\end{remark}
We note that for many practical values of the parameters (e.g.
$SNR(f)=80dB$, $d \leq 20$, $p \leq 100$) the first term in formula
(\ref{int main formula}), involving $2^{-2d}$, is dominant. Since we
are interested in results that have relevance to existing systems we
will develop in the next section, and under some further assumptions
(e.g. assumptions (\ref{assume_werner1}), (\ref{assume_werner2})), a
bound for $L_{i}(d)$ of the form $a_{1}2^{-2d}+a_{2}2^{-d}$ where
the coefficients $a_{1},a_{2}$ are expressible using the system
parameters. This is proposition \ref{correct-decay}.


{\bf Ensuring bounded transmission loss in each frequency bin}

We now turn to study the natural design requirement that the
transmission loss caused due to quantization of precoders should
be bounded by a certain fixed quantity, say 0.1bit/sec/Herz/user,
{\it on a per-tone basis}. Such a design criterion is examined in
the next corollary.
\begin{corollary}
Let $t>0$ and let $d$ be an integer with
\begin{equation}\label{bound on quantizer}
d \geq
d(t)
\end{equation}
With
\begin{displaymath}
d(t) = \left\{ \begin{array}{ll}
\log_{2}(1.25v(f)2^{t+1}/tln(2)) & \textrm{if $2^{t}-1 \leq \frac{B^{2}}{4A}$}\\
0.5\log_{2}(5(p-1)(1+r)^{2}SNR_{i}(f)/tln(2)) & \textrm{otherwise}
\end{array} \right.
\end{displaymath}

Then the transmission loss at tone $f$ due to quantization with
$d$ bits is at most $t$ bps/Hz.
\end{corollary}

{\it Proof:}
By theorem \ref{main theorem}, the loss at a tone $f$ is bounded
by
$\log_{2}\left(1+2^{-2d}u(f))\right)-\log_{2}((1-v(f)2^{-d})^{2}).$
Where $u(f)=2(p-1)(1+r)^{2}SNR_{i}(f)$ and
$v(f)=\sqrt{2}(1+r(f))$.

Using $1-2t \leq (1-t)^{2} $ we get

$$L_{i}(d,f) \leq \log_{2}\left(1+2^{-2d}u(f))\right)-\log_{2}(1-2v(f)2^{-d}).$$
We will show that the inequality
\begin{equation}\label{inq1}
\log_{2}\left(\frac{1+2^{-2d}u(f)}{1-2v(f)2^{-d}}\right) \leq t
\end{equation}
is satisfied for any $d \geq d(t)$ as in (\ref{bound on
quantizer}) Let $z=2^{-d}$ so that the inequality (\ref{inq1}) is
\begin{equation}
\frac{1+z^{2}u(f)}{1-2v(f)z} \leq 2^{t}
\end{equation}
This yields a quadratic inequality of the form
\begin{equation}
Az^{2}+Bz \leq T
\end{equation}
with $A=u(f)$, $B=2^{t+1}v(f)$ and $T=2^{t}-1$. Using lemma
\ref{quadratic lemma} (see section \ref{proof of corollaries} -
appendix D), we see that if $d \geq d_{0}(t)$ where
\begin{displaymath}
d_{0}(t) = \left\{ \begin{array}{ll}
\log_{2}(1.25v(f)2^{t+1}/(2^{t}-1) & \textrm{if $2^{t}-1 \leq \frac{B^{2}}{4A}$}\\
0.5\log_{2}(5(p-1)(1+r)^{2}SNR_{i}(f)/(2^{t}-1)) &
\textrm{otherwise}
\end{array} \right.
\end{displaymath}
Then $L_{i}(d,f) \leq t$. But $d_{0}(t) \leq d(t)$ because
$2^{t}-1 \geq ln(2)t$ and the result follows.



\begin{remark}
The qualitative behavior is $d(t) \approx a_{1}-\log_{2}(t)$ for
very small values of $t$ whereas $d(t) \approx a_{2}-0.5
\log_{2}(t)$ for
larger values of $t$. 

\end{remark}

\subsection{Applications of the Main Theorem}
We now apply theorem \ref{main theorem} to analyze the required
quantization level for DSM level 3 precoders under several design
criteria. To that end let $R_{i}$ be the transmission rate of the
$i$-th user (\ref{Transmission Rate}) and let $L_{i}$ be the
transmission loss of the $i$-th user as in (\ref{trans loss def}).
The {\it relative transmission loss} is defined by

\begin{equation}\label{trans loss}
\eta_{i}=\frac{L_{i}}{R_{i}}=\int_{f \in B} L_{i}(f)df / \int_{f
\in B} R_{i}(f)df
\end{equation}

The design criteria are

\begin{itemize}
\item Absolute/relative transmission loss across the band is
bounded.

\item Absolute/relative transmission loss for each tone is
bounded.

\end{itemize}

{\bf Bound on Absolute Transmission Loss}

From now on, we will assume that the transfer function obeys a
parametric model as in \cite{werner91}. Thus we assume
(\ref{assume_werner1}) and (\ref{assume_werner2}).


To bound the absolute transmission loss we estimate the integral
in formula (\ref{int main formula}) of theorem \ref{main theorem}.

Using the model (\ref{assume_werner1}) one can easily see that
$$SNR_{i}(f)=\frac{P_{i}(f)}{\sigma_{n_i}^{2}(f)}e^{-2 \alpha \ell \sqrt{f}}$$
Moreover, under the assumption (\ref{assume SLRD}) we have a
linear bound on the quantity $r(\mH(f,\ell))$ that is,



$$r(\mH(f,\ell)) \leq \gamma_{1}(\ell)+\gamma_{2}(\ell) f$$
Where $\gamma_{2})(\ell)=O(\sqrt{\ell})$. Putting these together we
can estimate the integral occurring in the bound (\ref{int main
formula}) and the final conclusion in described in theorem
\ref{correct-decay}.



The parameters $\gamma_{1}(\ell), \gamma_{2}(\ell)$ enter our
bounds through the following quantity.

\begin{equation}\label{rhoell}
\rho_{\ell}=(1+\gamma_{1}(\ell))^{2}+12(1+\gamma_{1}(\ell))
\frac{\gamma_{2}(\ell)}{(\alpha \ell)^{2}}+ 240
\left(\frac{\gamma_{2}(\ell)}{(\alpha \ell)^{2}})^{2} \right)
\end{equation}


\begin{remark}
The quantity $\rho_{\ell}$
 behaves as $1+C\ell^{-3/2}$ and is close to one for $\ell=300$m.
\end{remark}

We are now ready to formulate one of the main results of this paper:

\begin{theorem}\label{correct-decay}
Under assumptions {\bf Perfect CSI}, {\bf Quant$(2^{-d})$}, ${\bf
SPSD}$, {\bf Werner model} and {\bf sub-linear row dominance} (see
(\ref{assume Perfect CSI}), (\ref{assume Quantization}),
(\ref{assume SPSD}), (\ref{assume_werner1}), (\ref{assume_werner2}),
(\ref{assume SLRD})) we have
\begin{equation}
\frac{L_{i}(d)}{B} \leq \xi_{\ell} 2^{-2d}+2^{-d+3.5}
\end{equation}
where
\begin{equation}
\xi_{\ell}=\frac{4}{\ln(2)} (p-1) \frac{P}{\sigma_{n}^{2}}
\frac{1}{\alpha^{2}B} \frac{1}{\ell^{2}}\rho_{\ell}
%
\end{equation}
\end{theorem}

We provide a proof of this result in section
\ref{proof_correct_decay}(appendix E).

{\bf Bound on Relative Transmission Loss}

The most natural design criterion is to ensure that the {\it
relative capacity loss} is below a pre-determined threshold. We will
keep our assumption that the insertion loss behaves as in the model
(\ref{assume_werner1}), (\ref{assume_werner2}).


Let $SNR_{i}=\frac{P_{i}}{\sigma_{n_{i}}^{2}}$ and
$SNR'_{i}=\frac{P_{i}}{\sigma_{n_{i}}^{2}}e^{-\alpha \sqrt{B}}$ be
the Signal to Noise ratios of the $i$-th user at the lowest and
highest frequencies. We also denote by
$\widetilde{SNR}=\frac{SNR_{i}}{\Gamma}$ and by
$\widetilde{SNR'}_{i}=\frac{SNR'_{i}}{\Gamma}$. Finally, we denote
\begin{equation}\label{bnd for specteff}
c_{i}=\frac{1}{3} \log_{2}(\widetilde{SNR_{i}}) + \frac{2}{3}
\log_{2}(\widetilde{SNR'}_{i})
\end{equation}
The next proposition shows that $c_{i}$ provides a lower bound on
the spectral efficiency of the $i$-th user.

\begin{prop}\label{bound on spectral eff}
Assume that the attenuation transfer characteristic of the channel
is given by (\ref{assume_werner1}). Then the spectral efficiency is
bounded below by
\begin{equation}
\frac{1}{B}R_{i}
 \geq c_{i}
 \end{equation}
\end{prop}

The proof is deferred to section \ref{bound on spectral eff}
(appendix F).

\begin{corollary}\label{cor explicit bound}
Let $\eta_{i}(d)$ be the relative transmission rate loss of the
$i$-th user as in (\ref{trans loss}). Assume that the transfer
function satisfies  (\ref{assume_werner1}) and
(\ref{assume_werner2}).
 Then
\begin{equation}\label{explicit bound}
 \eta_{i}(d) \leq  \zeta_{\ell}2^{-2d}+\frac{1}{c_{i}}2^{-d+3.5}
\end{equation}
where
\begin{equation}\label{zeta}
\zeta_{\ell}=\frac{\xi_{\ell}}{c_{i}}=\frac{4}{\ln(2)} (p-1)
\frac{P}{\sigma_{n}^{2}} \frac{1}{\alpha^{2}B}
\frac{1}{\ell^{2}}\frac{1}{c_{i}} \rho_{\ell}.
\end{equation}

\end{corollary}

{\it Proof:}
This is an immediate consequence of the upper bound on the average
loss $\frac{L_{i}}{B}$ and the lower bound on $\frac{1}{B}R_{i}$.

{\bf Ensuring bounded relative transmission loss in the whole band}

The next corollary yields an upper bound for the number of quantized
bits required to ensure that the relative loss is below a given
threshold.
\begin{corollary}
Let $0 \leq \tau \leq 1$ and let $d \geq d(\tau)$ where
\begin{displaymath}
d(\tau) = \left\{ \begin{array}{ll}
\log_{2}(\frac{12 \sqrt{2}}{c_{i}\tau}) & \textrm{if $\tau \leq \frac{32}{  \zeta_{\ell} c^{2}   }$}\\
0.5\log_{2}(\frac{2.5}{\zeta_{\ell} \tau}) & \textrm{otherwise}
\end{array} \right.
\end{displaymath}
Then the relative transmission loss caused by quantization with
$d$ bits is at most $\tau$.
\end{corollary}
The proof is a simple application of the previous bound on the
relative transmission loss and lemma \ref{quadratic lemma} (section
\ref{proof of corollaries} - appendix D).

\section{Simulation Results}
To check the quality of the bounds in theorem \ref{main theorem}
and its corollaries, we compared the bounds with simulation
results, based on measured channels. We have used the results of
the measurement campaign conducted by France Telecom R$\&$D as
described in \cite{karipidis2006a}. All experiments used the band
$0-30$ MHz.
\subsection*{Full band}
For each experiment, we generated 1000 random precoder
quantization error matrices $\mE_2(f)$, with i.i.d. elements, and
independent real and imaginary parts, each uniformly distributed
in the interval $[-2^{-d},2^{d}]$.
We add the error matrix to the precoder matrix to generate the
quantized precoder matrix. Repeating this in each frequency we
produced a simulation of the quantized precoded system and
computed the resulting channel capacity of each of the 10 users.
Then we computed the relative and absolute capacity loss of each of
the users. In each bin we picked the worst case out of 1000
quantization trials and obtained a quantity we called {\em maximal
loss}. The quantity maximal loss is a random variable depending on
the number of bits used to quantize the precoder matrices. Each
value of this random variable provides a lower bound for the actual
worst case that can occur when the channel matrices are quantized.
We compare this lower bound with our upper bounds of theorem
\ref{main theorem}. We have checked our bounds in the following
scenario: Each user has flat PSD of -60dBm/Hz, the noise has flat
PSD of -140dBm/Hz. The Shannon Gap is assumed to be $10.7dB$. As can
be seen in figure 1, the bound given by (\ref{int main formula}) is
sharp. We also checked the more explicit bound (\ref{explicit
bound}) which is based on the model (\ref{assume_werner1}),
(\ref{assume_werner2}). We validated the linear behavior of the row
dominance $r(\mH(f))$ as a function of the tone $f$  as predicted by
formula (\ref{assume SLRD}).
Next we used (\ref{assume_werner1}) to fit the parameter $\alpha$ of
the cable via the measured insertion losses.  The process of fitting
is described in detail in \cite{karipidis2006a}. Its value
which was used in the bound (\ref{explicit bound}) was
$\alpha=0.0019$.
The parameters $\gamma_1=0.1596$ and $\gamma_2=3.1729 \ 10^{-8}$
were estimated from the measured channel matrices by simple line
fit. The results are depicted in Figure \ref{rel_cap_loss}.

\subsection*{Single frequency}

The bounds provided for the entire band are results of bounds on
each frequency bin. To show that our bounds are sharp even without
averaging over the frequency band, we studied the capacity loss in
specific frequency bins. We concentrated on the same scenario as
before (i.e. with 10 users), the noise is $-140dBm/Hz$ and the
power of the users is $-60dBm/Hz$. We picked measured matrices
$\mH(f_{1}), \mH(f_{2})$, so that $SNR(f_{1})$ is $40dBm$ and
$SNR(f_{2})$ is $60dBm$. As before, we systematically generated an
error matrix $\mE_{2}$ by choosing its entries to be i.i.d.,
uniformly distributed with maximal absolute value $2^{-d+0.5}$.
Next, we computed the transmission rate loss using formula
(\ref{trans loss formula}). By repeating this process $N=10000$
times
and choosing the worst event of transmission rate loss, we
obtained a lower bound estimate of worst-case transmission rate
loss. This was compared to the bounds of corollary \ref{cor one
tone}. The results are depicted in figure 2. Figure 2 uses formula
(\ref{one tone bound}). In particular we see that for $SNR=60dB$
and transmission rate loss of one percent, simulation indicates
quantization with 13 bits.
The
analytic formula indicates 14 bits. Similarly, when $SNR=40dB$, and
again allowing the same transmission rate loss of one percent,
simulation suggests using 10 bits for quantization. The simple
analytic estimate requires 11 bits.



\subsection*{The number of quantizer bits needed to assure 99 percent of capacity}
In the next experiment we have studied the number of bits required
to obtain a given transmission loss as a function of the loop
length. Figure 3
depicts the number of bits required to ensure transmission rate loss
below one percent as a function of loop length. We see that 14 bits
are sufficient for loop lengths up to 1200m. Fewer bits are required
for longer loops.

\subsection*{Stability of the results}
In the next experiment we validated that the analytic results proven
for perfect CSI are valid even when CSI is imperfect as long as
channel measurement errors are not the dominating cause for capacity
loss. To model the measurement errors of the channel matrix
$\mH(f)$, we used matrices with Gaussian entries with variance which
is proportional to $SNR(f)$. More precisely we assumed that the
estimation error of the matrix $\mH(f)$ is a Gaussian with zero mean
and with variance $\sigma_{\mH(f)}^{2}=\frac{1}{NSNR_{i}(f)}$, where
$N$ is the number of samples used to estimate the channel matrix
$\mH(f)$. For $N=1000$, we estimated the loss in a frequency bin as
the worst case out of 500 realizations of quantization noise
combined with measurement noise. Figure 3 shows that as long as the
quantization noise is dominant we can safely use our bounds for the
transmission loss. We comment that the stationarity of DSL channels
allows accurate channel estimation.

\section{Conclusions}
In this paper we analyzed finite word length effects on the
achievable rate of vector DSL systems with zero forcing precoding.
The results of this paper provide simple analytic expressions for
the loss due to finite word length. These expressions allow simple
optimization of linearly precoded DSM level 3 systems.

We validated our results using measured channels. Moreover, we
showed that our bounds can be adapted to study the effect of
measurement errors on the transmission loss. In practice for loop
lengths between 300 and 1200 meters, one needs 14 bits to
represent the precoder elements in order to lose no more than one
percent of the capacity.

\section{Appendix A: Proof of Lemma \ref{lemma of section 3}}\label{proof of
lemma 1} In this section we prove lemma \ref{lemma of section 3}.
%

{\it Proof:}
For simplicity we will omit the explicit dependency of the
matrices $\mH(f),\mD(f),\mF(f),\mP(f)$ on the frequency $f$. We
show that
\begin{equation}
\mH \mP = \mD + \mD \mgD,
\end{equation}
with $\mgD$ as above. Indeed $\mH= \mD (\mI+\mD^{-1}\mF)$ and thus
\begin{equation}
\mH \mP=\mD (\mI+{\mD}^{-1}{\mF})
((\mI+\mD^{-1}\mF+\mE_{1})^{-1}+\mE_{2}).
\end{equation}
Hence,
\begin{equation}
\mH \mP= \mD (\mI+{\mD}^{-1}{\mF}+\mE_{1}-\mE_{1})
(\mI+\mD^{-1}\mF +\mE_{1})^{-1}+ \mD (\mI+{\mD}^{-1}{\mF})
\mE_{2}.
\end{equation}
Thus,
\begin{equation}
\mH \mP= \mD - \mD \mE_{1} (\mI+\mD^{-1}\mF +\mE_{1})^{-1}+ \mD
(\mI+{\mD}^{-1}{\mF}) \mE_{2},
\end{equation}
Which proves the lemma.

\section{Appendix B: Proof of Lemma \ref{main prop}}\label{proof of
lemma 2}

In this appendix we prove lemma \ref{main prop}.






{\it Proof:}
By equation (\ref{signal model real}), the $i$-th user receives
\begin{equation}\label{system with errors}
x_{i}(f)=d_{i,i}(f)s_{i}(f)+d_{i,i} \sum_{j=1}^{p}
\Delta_{i,j}(f)s_{j}(f)+n_{i}(f)=d_{i,i}(f)(1+\Delta_{i,i}(f))s_{i}(f)+N_{i}(f)
\end{equation}
with $N_{i}(f)=d_{i,i}(f) \sum_{j \ne i}^{p}
\Delta_{i,j}(f)s_{j}(f)+n_{i}(f)$. For a large number of users, we
may assume that $N_{i}(f)$ is again a Gaussian noise and the {\it
transmission rate} at frequency $f$ of the system described by
equation (\ref{system with errors}) will be
\begin{equation}
R_{i}(\Delta,f)= \log_{2}\left (1+ \frac{
P_{i}(f)|d_{i,i}(f)|^{2}|(1+\Delta_{i,i}(f))|^{2} }{\Gamma(\sum_{j
\ne i}
P_{j}(f)|d_{i,i}(f)|^{2}|\Delta_{i,j}(f)|^{2}+|n_{i}(f)|^{2})}
\right)
\end{equation}
Note that this quantity appeared in the main body of the paper
just after equation (\ref{R tilde}) where it was denoted
$\tilde{R}_{i}(f)$. Dividing both the numerator and denominator by
$P_{i}(f)|d_{i,i}(f)|^{2}$ we get
\begin{equation}
R_{i}(\Delta,f)= \log_{2}\left (1+ \frac{
|(1+\Delta_{i,i}(f))|^{2} }{\Gamma \sum_{j \ne i}
\frac{P_{j}(f)}{P_{i}(f)} |\Delta_{i,j}(f)|^{2}+ \frac{\Gamma
|n_{i}(f)|^{2}}{P_{i}(f)|d_{i,i}(f)|^{2}}} \right)
\end{equation}
or
\begin{equation}
R_{i}(\Delta,f)= \log_{2}\left (1+ \frac{
|(1+\Delta_{i,i}(f))|^{2} }{ \delta_{i}(f)+ \frac{1}{eSNR_{i}(f)}}
\right)
\end{equation}
where we have defined \beq
eSNR_{i}(f)=\frac{SNR_{i}(f)}{\Gamma}=\frac{P_{i}(f)|d_{i,i}(f)|^{2}}{\Gamma
|n_{i}(f)|^{2}} \eeq and
\begin{equation}\label{deltanorm}
\delta_{i}(f)=\Gamma \sum_{j \ne i} \frac{P_{j}(f)}{P_{i}(f)}
|\Delta_{i,j}(f)|^{2}
\end{equation}
To get the transmission rate loss we denote
\begin{equation}
eSNR_{i}(\Delta,f)=\frac{|(1+\Delta_{i,i}(f))|^{2}}{\delta_{i}(f)
+\frac{1}{eSNR_{i}(f)}} \end{equation}

Notice that $$eSNR_{i}(f)=eSNR_{i}({\bf 0},f)$$

By definition (\ref{trans loss def}) we have
\begin{equation}
L_{i}(\mgD,f)=R_{i}(f)-R_{i}(\Delta,f)=\log_{2}(1+eSNR_{i}(f))-\log_{2}(1+eSNR_{i}(\Delta,f))
\end{equation}
We then have
\begin{equation}
L_{i}(\mgD,f)=-\log_{2}\left(\frac{1+eSNR_{i}(\Delta,f)}{1+eSNR_{i}(f)}\right)=-\log_{2}\left(1-
\frac{eSNR_{i}(f)-eSNR_{i}(\Delta,f)}{1+eSNR_{i}(f)}\right)
\end{equation}
But
\begin{equation}
eSNR_{i}(f)-eSNR_{i}(\Delta,f)=eSNR_{i}(f)-\frac{|(1+\Delta_{i,i}(f))|^{2}}{\delta_{i}
+\frac{1}{eSNR_{i}(f)}}
\end{equation} so
\begin{equation}
eSNR_{i}(f)-eSNR_{i}(\Delta,f)=\frac{eSNR_{i}(f)\delta_{i}(f)+1-|(1+\Delta_{i,i}(f))|^{2}}{\delta_{i}
+\frac{1}{eSNR_{i}(f)}}
\end{equation}
and finally,
\begin{equation}
eSNR_{i}(f)-eSNR_{i}(\Delta,f)=eSNR_{i}(f)\frac{eSNR_{i}(f)\delta_{i}(f)+1-|(1+\Delta_{i,i}(f))|^{2}}{\delta_{i}(f)eSNR_{i}(f)+1}
\end{equation}
Hence
\begin{equation}\label{trans loss formula1}
L_{i}(\mgD,f)=-\log_{2}\left(1-\frac{eSNR_{i}(f)}{eSNR_{i}(f)+1}
\frac{a_{i}(f)+1-|1+\Delta_{i,i}|^{2}}{a_{i}(f)+1}\right)
\end{equation}
where \beq a_{i}(f)=\delta_{i}(f)eSNR_{i}(f) \eeq and
$\delta_{i}(f)$ is given in (\ref{deltanorm}). With the notations
(\ref{q number}) and (\ref{k number}) we get the formula
\begin{equation}\label{trans loss formula2}
L_{i}(\mgD,f)=-\log_{2}\left(1-k_{i}(f)(1-q_{i}(\mgD,f))\right)
\end{equation}

To prove the bound we consider two cases. When $q(\mgD,f)
> 1$ we see from equation (\ref{trans loss formula2}) that $L_{i}(\mgD,f) \leq 0$.
This clearly indicates transmission gain and the stated inequality
is valid. On the other hand, if $q_{i}(\mgD,f) \leq 1$ we get

 \beq \frac{eSNR_{i}}{eSNR_{i}+1}(1-q_{i}(\mgD,f)) \leq 1-q_{i}(\mgD,f)
\eeq and using the monotonicity of $-\log_{2}(1-u)$ (increasing)
in the interval $(0,1)$, we get \beq L_{i}(\mgD,f) \leq
-\log_{2}\left(1-(1-q_{i}(\mgD,f)))\right)=\log_{2}\left(\frac{1}{q_{i}(\mgD,f)}\right)
\eeq and the Lemma is proved.

\section{Appendix C: Proof of theorem \ref{main theorem}}\label{Appendix proof
of main theorem}

For the proof of the theorem we need a simple lemma.
\begin{lemma}\label{simple lemma}
Let $\mA$ be a complex $p \times p$ matrix and define $\mD$ to be
the diagonal matrix with $\mD_{i,i}=\mA_{i,i}$ for $i=1,..,p$. Let
$\mE$ be a $p \times p$ matrix whose entries are complex numbers
with real and imaginary parts bounded by $2^{-d}$. Finally, let
$\mB=\mD^{-1}\mA \mE$. Then $|\mB_{i,j}| \leq 2^{-d+1/2}(1+r(A))$.
\end{lemma}

{\it Proof:} Let $\mQ=\mD^{-1}\mA=\mI+\mD^{-1}(\mA-\mD)$. Then we
have \beq \sum_{k=1}^p|\mQ_{ik}| \le 1+r(\mA) \eeq for all
$i=1,..,p$.
Therefore \beq |\mB_{i,j}|=\left|\sum_{k=1}^p \mQ_{ik} \mE_{kj}
\right| \leq 2^{-d+1/2}\sum_{k=1}^p|Q_{ik}| \leq 2^{-d+1/2}(1+r)
\eeq


{\bf Proof of the main theorem}

We first bound the loss $L_{i}(f)$ in a particular tone $f$. By
Lemma \ref{main prop} we have

\begin{equation}
L_{i}(\mgD,f) \leq
Max\left(0,\log_{2}\left(\frac{1}{q_{i}(\mgD,f)}\right)\right)
\end{equation}
where
\begin{equation}\label{def of q}
q_{i}(\mgD,f)=\frac{|1+\Delta_{i,i}(f)|^{2}}{a_{i}(f)+1}
\end{equation}

Here $\mgD(f)=(\mI+\mD(f)^{-1}\mF(f))\mE_{2}(f)$ where
$\mH(f)=\mD(f)+\mF(f)$ is the channel matrix at frequency $f$ and
$\mE_{2}(f)$ is a matrix whose entries are complex numbers with real
and imaginary parts bounded by $2^{-d}$. Applying Lemma (\ref{simple
lemma}) to the matrix $\mH(f)$ we see that the entries
$\mgD_{i,j}(f)$ are all in a disk of radius $v(f)2^{-d}$ around
zero. Using $r(f) \leq 5$ we obtain $v(f)=\sqrt{2}(1+r(f)) \leq 6
\sqrt{2}$. Using $d \geq 4$ we get $1-2^{-d}v(f) \geq 1-
\frac{6\sqrt{2}}{16} > 0$.

Thus
\begin{equation}\label{bound on diagonal}
|1+\mgD_{i,i}(f)|^{2} \geq (1-v2^{-d})^{2}.
\end{equation}

Using the assumption on the PSD of the different users we obtain

\begin{equation}\label{use of SPSD}
a_{i}(f)=\sum_{j \ne i} \frac{P_{j}(f)}{P_{i}(f)}
|\Delta_{i,j}(f)|^{2} SNR_{i}(f)=\sum_{j \ne i}
|\Delta_{i,j}(f)|^{2} SNR_{i}(f).
\end{equation}

Using Lemma (\ref{simple lemma}) we have

\begin{equation}\label{bound on delta norm}
 \sum_{j \ne i} |\Delta_{i,j}(f)|^{2} \leq
(p-1)2^{-2d+1}(1+r(f))^{2}, \end{equation}

thus,

\begin{equation}\label{bound on ai}
1+a_{i}(f)  \leq 1+(p-1)2^{-2d+1}(1+r(f))^{2}SNR_{i}(f) =
1+\gamma(d,f)SNR_{i}(f).
\end{equation}

Combining (\ref{def of q}), (\ref{bound on diagonal}) and
(\ref{bound on ai}) we obtain

\begin{equation}\label{local bound} \frac{1}{q_{i}(\mgD,f)} \leq
\frac{1+\gamma(d,f)SNR_{i}(f)}{(1-v(f)2^{-d})^{2}} \end{equation}

Note that the right hand side of the above inequality is positive
and greater than one. Combining (\ref{basic inq}) and (\ref{local
bound}) we obtain

\begin{equation}\label{main inequality}
L_{i}(\mgD,f) \leq
 \log_{2}\left(\frac{1+\gamma(d,f)SNR_{i}(f)}{(1-v(f)2^{-d})^{2}}\right)
= \log_{2}(1+\gamma(d,f)SNR_{i}(f))
 -2\log_{2}(1-v(f)2^{-d})
\end{equation}

Since $\gamma(d,f) \leq 2(1+r_{max})^{2}(p-1)2^{-2d}$ and
$v(f)=\sqrt{2}(1+r(f)) \leq \sqrt{2}(1+r_{max})$, integrating this
inequality over $f \in B$ we obtain (\ref{int main formula}) and
the theorem is proved.


\section{Appendix D: Proofs of Corollary 4.8 and 4.9}\label{proof of corollaries}

\subsection{A Quadratic Inequality}
In the proof of corollary 4.8 and corollary 4.9 we use the following
lemma.
\begin{lemma}\label{quadratic lemma}
Let $A,B,T$ be positive real numbers and let
\begin{displaymath}
d(T) = \left\{ \begin{array}{ll}
\log_{2}(1.25B/T) & \textrm{if $T \leq \frac{B^{2}}{4A}$}\\
0.5\log_{2}(2.5A/T) & \textrm{otherwise}
\end{array} \right.
\end{displaymath}
Then for $d \geq d(T)$ we have \beq A2^{-2d}+B2^{-d} \leq T \eeq
\end{lemma}

{\it Proof:} We let $x=2^{-d}$ and observe that $f(x)=Ax^{2}+Bx$ is
monotone in $x>0$ with one root of $f(x)=T$ exactly at
$x_{0}=\frac{\sqrt{B^{2}+4AT}-B}{2A}$. Thus for any $d
> d_{0}(T)=\log_{2}(\frac{2A}{\sqrt{B^{2}+4AT}-B})$ we have
$A2^{-2d}+B2^{-d}=f(2^{-d}) \leq f(2^{-d_{0}})=f(x_{0})=T.$ To
complete the proof we will show that $d_{0}(T) \leq d(T)$. Indeed,
\beq
d_{0}(T)=\log_{2}\left(\frac{2A}{\sqrt{B^{2}+4AT}-B}\right)=\log_{2}\left(\frac{2A(\sqrt{B^{2}+4AT}+B)}{4AT}\right)
\eeq Thus,

\beq
d_{0}(T)=\log_{2}\left(\frac{B}{2T}(\sqrt{1+\frac{4AT}{B}}+1)\right)
\eeq If we let $\rho=\frac{4AT}{B^{2}}$ then for $\rho <1 $ we have
$\sqrt{1+\rho}+1 \leq 2.5$ and this yields the bound \beq d_{0}(T)
\leq \log_{2}\left(\frac{1.25B}{T}\right) \eeq for $T \leq
\frac{B^{2}}{4A}$. On the other hand if $\rho > 1$ it is easy to see
that $1+\sqrt{1+\rho}
 \leq 2.5 \sqrt{\rho}$ thus

\beq d_{0}(T) \leq \log_{2}\left(\frac{B}{2T}(2.5
\sqrt{\frac{4AT}{B^{2}}})\right)= \log_{2}\left(2.5
\sqrt{\frac{A}{T}}\right)
 \eeq

\begin{remark}
Note that as $T$ decreases to zero the value of $d(T)$ increases and
behaves as $\log_{2}(\frac{1}{T})$.
\end{remark}

\section{Appendix E: proof of theorem \ref{correct-decay}}
\label{proof_correct_decay} {\it Proof:}
Using Theorem \ref{main theorem}, the capacity loss of the $i$-th
user, $L_{i}(d)$,  is bounded by
\begin{equation}\label{step one}
 L_{i}(d) \leq \int_{f \in B}
\log_{2}(1+\gamma(d,f)\frac{P}{\sigma_{n}^{2}}e^{-\alpha_{\ell}{\sqrt{f}}})df
- 2|B|\log_{2}(1-2^{-d+1.5}) \end{equation} By assumption,
$\gamma(d,f) \leq 2(p-1)2^{-2d}(1+\gamma_{1}+\gamma_{2} f)^{2}$. To
bound the first term we state here a simple lemma (for the proof see
section \ref{J integral appendix} - appendix G)).
\begin{lemma}\label{J integral}
Let $f(x)=\frac{P}{\sigma_{n}^{2}}e^{-\alpha \sqrt{x}}$ and define
\beq J_{a,b}(\mu)=\frac{1}{B}\int_{0}^{B} \log_{2}(1+\mu
(a+bx)^{2}f(x))dx \eeq We have \beq J(\mu) \leq \frac{e^{\alpha
\sqrt{B}}}{\alpha^{2}B}\left(2a^{2}+24\frac{ab}{\alpha^{2}} + 240
\left(\frac{a}{\alpha^{2}}\right)^{2}\right) \log_{2}\left(1+\mu
f(B)\right) \eeq
\end{lemma}

We can now finish the proof of the theorem.

Let $a=1+\gamma_{1}$, $b=\gamma_{2}$ and $\mu=2(p-1)2^{-2d}$, and
let $J=J_{a,b}$ as in the lemma above. From (\ref{step one}) we
get

\beq \frac{1}{B}L_{i}(d) \leq J(2(p-1)2^{-2d}) -
2\log_{2}(1-2^{-d+1.5}) \eeq

Using the inequality $-\log_{2}(1-z) \leq 2 z$, for $z \leq
\frac{1}{2}$, and the inequality provided by the lemma for $J(\mu)$
we obtain \beq \frac{1}{B}L_{i}(d) \leq \frac{e^{\alpha \ell
\sqrt{B}}}{\alpha^{2}B}\left(2(1+\gamma_{1}(\ell))^{2}+24(1+\gamma_{1}(\ell))
\frac{\gamma_{2}(\ell)}{(\alpha \ell)^{2}}+ 240
\left(\frac{\gamma_{2}(\ell)}{(\alpha \ell)^{2}}\right)^{2} \right)
\log_{2}(1+2(p-1)2^{-2d} f(B)) + 2^{-d+3.5} \eeq

Using $\log_{2}(1+t) \leq ln(2) t$, the fact that
$f(B)=\frac{P}{\sigma_{n}^{2}}e^{-\alpha \sqrt{B}}$ and the
definition of $\rho_{\ell}$ in (\ref{rhoell}) we obtain

\beq \frac{1}{B}L_{i}(d) \leq \frac{4}{ln(2)}  (p-1)
\frac{P}{\sigma_{n}^{2}} \frac{1}{(\alpha \ell)^{2}B}\rho_{\ell}
2^{-2d} + 2^{-d+3.5} \eeq

\section{Appendix F: Proof of proposition \ref{bound on spectral
eff}}

{\it Proof:}
We begin with a bound on the transmission rate of the users. By
(\ref{Transmission Rate}) and the model (\ref{assume_werner1}) we
obtain
\begin{equation}
R_{i}=\int_{f \in B} \log_{2}(1+\Gamma^{-1}SNRe^{-\alpha_{\ell}
\sqrt{f}})df \geq \log_{2}(e) \int_{f \in B}
\ln(\Gamma^{-1}SNRe^{-\alpha_{\ell} \sqrt{f}})df
\end{equation}
Thus,

\begin{equation}
R_{i} \geq  B \log_{2}(\Gamma^{-1}SNR)-\log_{2}(e) \int_{0}^{B}
\alpha_{\ell} \sqrt{f}df \geq B \log_{2}(\Gamma^{-1}SNR)-
\frac{2}{3}\log_{2}(e) \alpha_{\ell} B \sqrt{B}
\end{equation}
We notice that this, with $SNR'=SNRe^{-\alpha_{\ell} \sqrt{B}}$
implies
\begin{equation}
\frac{1}{B} R_{i} \geq \log_{2}(\Gamma^{-1}SNR)- \frac{2}{3}
\log_{2}(e)
(\ln(SNR)-\ln(SNR'))=\frac{1}{3}\log_{2}(SNR)+\frac{2}{3}\log_{2}(SNR')-\log_{2}(\Gamma),
\end{equation}
and the proof is complete.
\begin{remark}
In practice, the estimation of $\alpha_{\ell}$ is more reliable than
the measurement of the transfer function at the edge of the
frequency band. Thus, the equivalent form  \beq \frac{1}{B}R_{i}
 \geq \log_{2}(\widetilde{SNR}) - \frac{2}{3} \alpha_{\ell} \sqrt{B}
\eeq is more reliable.
\end{remark}

\section{Appendix G: Proof of lemma \ref{J integral}}\label{J
integral appendix}

In this section we prove lemma \ref{J integral}. Recall


\beq J(\mu)=\frac{1}{B}\int_{0}^{B} \log_{2}(1+\mu (a+bx)^{2}f(x))dx
\eeq where $f(x)=\frac{P}{\sigma_{n}^{2}}e^{-\alpha \sqrt{x}}$

{\it Lemma:} Let $a \geq 1$ and $b \geq 0$.
Let $M$ be the maximal value of $(a+bx)^{2}f(x)$ in the interval
$[0,B]$. We have \beq J(\mu) \leq min\left(\frac{e^{\alpha
\sqrt{B}}}{\alpha^{2}B}\left(2a^{2}+24\frac{ab}{\alpha^{2}}+ 240
(\frac{b}{\alpha^{2}})^{2}\right) \log_{2}\left(1+\mu
f(B)\right),\log_{2}\left(1+M \mu\right)\right) \eeq

In particular we have \beq J(\mu) \leq \frac{2P}{ln(2)
\alpha^{2}B\sigma_{n}^{2}}\left(a^{2}+12\frac{ab}{\alpha^{2}}+ 120
\left(\frac{b}{\alpha^{2}}\right)^{2}\right) \mu  \eeq
 which is sharp for
small values of $\mu$.

{\it Proof:}
The inequality $J(\mu) \leq \log_{2}(1+M \mu)$ is evident. To get
the second bound we compute the derivative with respect to $\mu$
\begin{equation}\label{der of J}
 J^{'}(\mu) = \frac{1}{B ln(2)}\int_{0}^{B}
\frac{(a+bx)^{2}f(x)}{1+\mu (a+bx)^{2}f(x)}dx \end{equation} Using
the lower bound $1+ \mu (a+bx)^{2}f(x) \geq 1+ \mu f(x) \geq 1+f(B)
\mu$ we obtain

 \beq J^{'}(\mu) \leq \frac{1}{B ln(2)}\int_{0}^{B}
\frac{(a+bx)^{2}f(x)}{1+\mu f(B)}dx \eeq

We get \beq J^{'}(\mu) \leq \frac{P}{\sigma_{n}^{2}}\frac{1}{B
ln(2)}\frac{1}{(1+\mu f(B))}\int_{0}^{B} (a^{2}+2abx+ b^{2}x^{2})
e^{-\alpha \sqrt{x}}dx \eeq

But $\int_{0}^{\infty} x^{n}e^{-\sqrt{x}}dx=2\int_{0}^{\infty}
t^{2n+1}e^{-t}dt=2(2n+1)!$ and hence

\beq \int_{0}^{B} e^{-\alpha \sqrt{x}}dx \leq \frac{1}{\alpha^{2}}
\int_{0}^{\infty} e^{-\sqrt{x}}dx = \frac{2}{\alpha^{2}}\eeq

\beq \int_{0}^{B} x e^{-\alpha \sqrt{x}}dx \leq \frac{1}{\alpha^{4}}
\int_{0}^{\infty} x e^{-\sqrt{x}}dx = \frac{12}{\alpha^{4}} \eeq

\beq \int_{0}^{B} x^{2} e^{-\alpha \sqrt{x}}dx \leq
\frac{1}{\alpha^{6}} \int_{0}^{\infty} x^{2} e^{-\sqrt{x}}dx =
\frac{240}{\alpha^{6}}  \eeq

Thus we get \beq J^{'}(\mu) \leq \frac{P}{\sigma_{n}^{2}}\frac{1}{B
ln(2)}\frac{1}{(1+\mu
f(B))}\left[\frac{2a^{2}}{\alpha^{2}}+\frac{24ab}{\alpha^{4}}+
\frac{240 b^{2}}{\alpha^{6}}\right]\eeq





Integrating this inequality from $\mu=0$ to $t$ we obtain

\beq \int_{0}^{t} J^{'}(\mu) \leq
\frac{2}{ln(2)}\frac{P}{\sigma_{n}^{2}}\frac{1}{\alpha^{2}B}
\frac{\ln(1+t f(B))}{f(B)} \left[a^{2}+\frac{12ab}{\alpha^{2}}+
\frac{120 b^{2}}{\alpha^{4}}\right]\eeq Using the fact that
$J(0)=0$, we obtain the desired result.

\begin{remark}
We emphasize that $M$ can be computed analytically. In fact, it is a
routine exercise to write the maxima $M$ of $f(x)$ in terms of
$a,b,\alpha$.  Indeed
$$f'(x)=\frac{P}{\sigma_{n}^{2}}(2b(a+bx)e^{-\alpha \sqrt{x}}-\frac{(a+bx)^{2}}{2 \alpha
\sqrt{x}}^{-\alpha \sqrt{x}}).$$ Thus $f'(x)=0$ is equivalent to a
quadratic equation, and can be solved analytically. Since the
function $f(x)$ may have at most two critical point, say $x_{1},
x_{2} \in[0,\infty)$ we find that
$$M=max(f(0),f(x_{1}),f(x_{1}),f(B)).$$

\end{remark}

\section{Appendix H: Lifting the assumption of equal PSD from the main
theorem}\label{lifting of P}


In this appendix we prove a slight generalization of the main
result, showing that the assumption of equal PSD in the binder is
not necessary. The resulting bound is similar to that of the main
theorem \ref{main theorem}.

To formulate the bound on the transmission loss we introduce the
quantities

\begin{equation}
P_{max}(f)=max_{i}(P_{i}(f))
\end{equation}

\begin{equation}
P_{min}(f)=min_{i: P_{i}(f) \ne 0}(P_{i}(f))
\end{equation}

We let $\rho(f)=P_{max}(f)/P_{i}(f)$. We will say that the PSD
satisfies the assumption {\bf SPSD($\rho$)} (or has dynamic range of
width $\rho$) if we have
$$P_{max}(f) \leq \rho P_{min}(f)$$
We emphasize that this means that for each $f$ such that $P_{i}(f)
\ne 0$ we have
$$P_{max}(f) \leq \rho P_{i}(f)$$

\begin{remark}
In realistic scenarios the number $\rho$ is limited by the maximal
power back-off parameter of the modems in the system.
\end{remark}

\begin{theorem} \label{generalization of the main theorem}
Assume assumptions {\bf Perfect CSI}, {\bf Quant$(2^{-d})$}, and
{\bf SPSD($\rho$)}. Assume that the precoder $\mP(f)$ is quantized
using $d \geq \frac{1}{2}+\log_2(1+r_{\max}))$ bits. Let $\mH(f)$
be the channel matrix of $p$ twisted pairs at frequency $f$. Let
$r(f)=r(\mH(f))$ as in (\ref{row_dominance_parameter}). The
transmission rate loss of the $i$-th user at frequency $f$ due to
quantization is bounded by
\begin{equation}\label{generalized bnd} L_{i}(d,f) \leq \log_{2}(1+\gamma(d,f)SNR_{i}(f))-
2\log_{2}(1-v(f)2^{-d}), \end{equation} where \beq \gamma(d,f)=2
\rho(f) (p-1)(1+r(f))^{2}2^{-2d}
 \eeq and
\beq  v(f)=\sqrt{2}(1+r(f)). \eeq Furthermore, the transmission loss
in the band $B$ is at most
\begin{equation}\label{generalized main formula}
\int_{f \in B} \log_{2}(1+\gamma(d)SNR_{i}(f))df - 2|B|
\log_{2}(1-(1+r_{max})2^{-d+0.5}),
\end{equation}
where $|B|$ is the total bandwidth, and
\beq \gamma(d)=2 \rho
(1+r_{max})^{2}(p-1)2^{-2d}.
 \eeq
\end{theorem}

\begin{proof}
Only few changes in the proof of theorem \ref{main theorem} are
needed in order to derive the above theorem. In the proof of the
main theorem instead of (\ref{use of SPSD}) we have
\begin{equation}\label{reuse of SPSD}
a_{i}(f)=\sum_{j \ne i} \frac{P_{j}(f)}{P_{i}(f)}
|\Delta_{i,j}(f)|^{2} SNR_{i}(f) \leq \sum_{j \ne i} \rho(f)
|\Delta_{i,j}(f)|^{2} SNR_{i}(f).
\end{equation}
The bound on $\Delta_{i,j}(f)$ obtained in (\ref{bound on delta
norm}) is valid because our assumptions on the quantization are
the same as in theorem \ref{main theorem}. Following the same line
of reasoning as in equations (\ref{bound on ai})-(\ref{local
bound}) yields the bound (\ref{generalized bnd}). This, together
with the assumption {\bf SPSD($\rho$)} easily yields
(\ref{generalized main formula}).

\end{proof}


\newpage

\begin{figure}
\centering
\includegraphics[width=8cm]{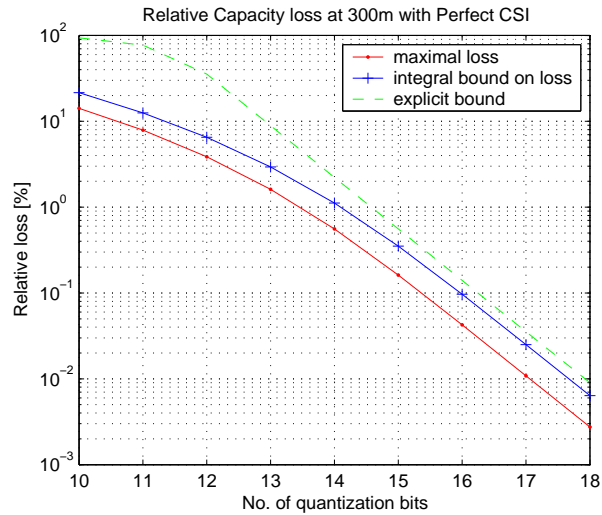}
\caption{Relative Capacity loss vs. number of quantizer bits in
perfect CSI in a system of 10 users. Integral bound on loss is
obtained via equation (\ref{int main formula}), explicit bound is
obtained via (\ref{explicit bound}) and equations (\ref{zeta}),
(\ref{bnd for specteff}).}
\label{rel_cap_loss}
\end{figure}





\begin{figure}
\centering
\includegraphics[width=8cm]{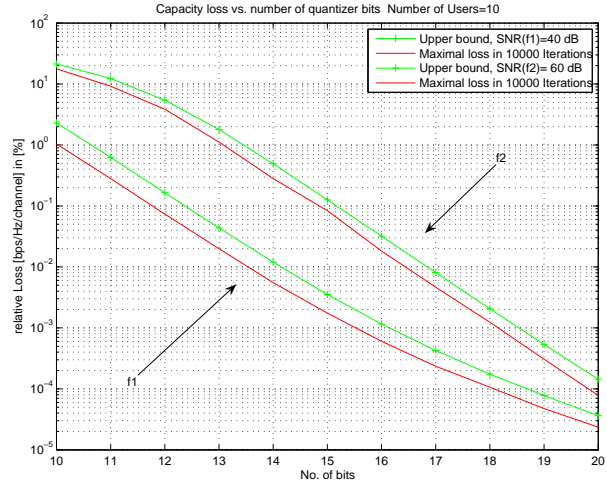}
\caption{Capacity loss vs. quantizer bits. Perfect CSI in system
of 10 users}
\end{figure}

\begin{figure}\label{fig_loop length versus quantization level}
\centering
\includegraphics[width=8cm]{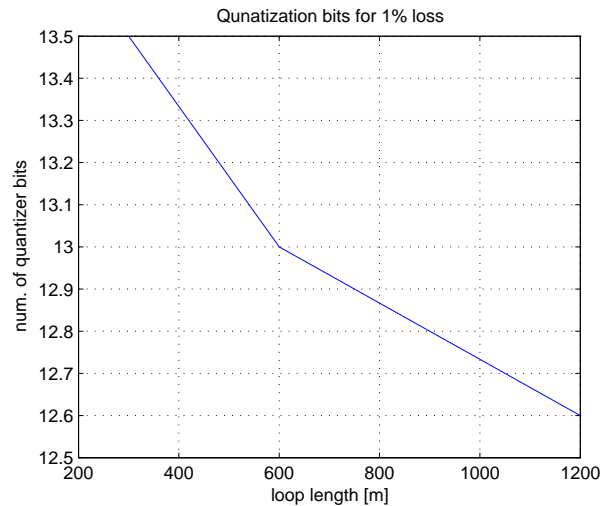}
\caption{Number of quantization bits required vs. loop length}
\end{figure}

\begin{figure}
\centering
\includegraphics[width=8cm]{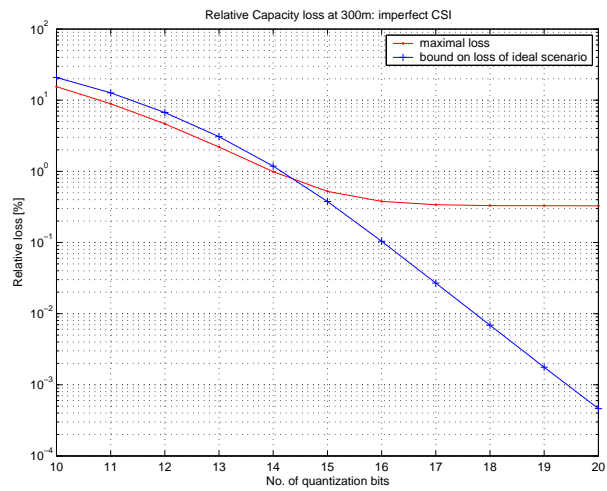}
\caption{Capacity loss vs. quantizer bits. Imperfect CSI in system
of 10 users. CSI based on 1000 measurements}
\end{figure}

\end{document}